\documentclass[11pt,reqno]{amsart}

\usepackage{amssymb}

\usepackage{amsfonts}

\usepackage{amsmath}

\usepackage[all]{xy}

\newtheorem{theorem}{Theorem}[section]

\newtheorem{lemma}[theorem]{Lemma}

\newtheorem{proposition}[theorem]{Proposition}

\newtheorem{Definition}[theorem]{Definition}

\newtheorem{Example}[theorem]{Example}

\newtheorem{Remark}[theorem]{Remark}

\newenvironment{remark}{\begin{Remark}\begin{em}}{\end{em}\end{Remark}}

\newenvironment{definition}{\begin{Definition}\begin{em}}{\end{em}\end{Definition}}

\newcommand{\B}{\mathbf{B}}

\newcommand{\R}{{\mathbb R}}

\newcommand{\D}{{\mathbb D}}

\newcommand{\uu}{\mathbf{u}}

\newcommand{\vv}{\mathbf{v}}

\newcommand{\ww}{\mathbf{w}}

\newcommand{\0}{\mathbf{0}}

\DeclareMathOperator{\gyr}{gyr}

\DeclareMathOperator{\tr}{tr}

\address{Department of Mathematics,
Louisiana State University, Baton Rouge, LA 70803, U.S.A.}

\begin{document}

\author{Sejong Kim}

\title{Distances of qubit density matrices on Bloch sphere}


\maketitle

We recall the Einstein velocity addition on the open unit ball $\B$
of $\R^{3}$ and its algebraic structure, called the Einstein
gyrogroup. We establish an isomorphism between the Einstein
gyrogroup on $\B$ and the set of all qubit density matrices
representing mixed states endowed with an appropriate addition. Our
main result establishes a relation between the trace metric for the
qubit density matrices and the rapidity metric on the Einstein
gyrogroup $\B$.

\vspace{5mm}

\noindent {\bf PACS} (2010):  03.65.Fd, 03.67.-a

\noindent {\bf Keywords}: Einstein addition, gyrogroup, qubit,
rapidity metric, trace metric

\section{Introduction}

Einstein addition is the standard velocity addition of
relativistically admissible vectors that Einstein introduced in his
1905 paper that founded the special theory of relativity. In his
book \cite{Ein} the presentation of Einstein's special theory of
relativity is based on Einstein velocity addition law. It also
allows the reader to take the properties of the underlying
hyperbolic geometry.

A. A. Ungar has introduced and studied in \cite{Un08} gyrogroups
that are algebraic settings of hyperbolic geometry, and suggested
three examples of gyrogroups corresponding to three models of
hyperbolic geometry. It has been known that gyrogroups correspond
equivalently to loop structures; see \cite{SSS}. In Section 2 we
review the Einstein gyrogroup on the open unit ball of the Euclidean
three-dimensional space $\R^{3}$ corresponding to the Beltrami-Klein
ball model of hyperbolic geometry.

Bloch vectors in the unit ball of $\R^{3}$ are well-known in quantum
information and computation theory. A qubit density matrix, a
$2$-by-$2$ positive semidefinite Hermitian matrix with trace $1$, is
completely described by the Bloch vectors. Chen and Ungar have
computed the Bures fidelity of qubit density matrices generated by
two Bloch vectors and showed some equivalent formulas in their
papers \cite{CU1} and \cite{CU2}. The Bures fidelity that plays an
important role for the geometry of quantum state space has been
manipulated into a form possessing distance properties, called the
Bures metric. On the other hand, we apply the trace metric to the
space of all invertible qubit density matrices, investigate its
properties, and show in Section 5 the relation between the rapidity
metric of the Einstein gyrogroup on $\B$ and the trace metric of the
qubit density matrices.

\section{Einstein addition and gyrogroups}

The velocities in Einstein's theory of special relativity are
three-dimensional vectors with magnitude bounded by the speed of
light. We assume the speed of light is normalized by the value 1,
and call such velocities admissible vectors. The relativistic sum of
two admissible vectors $\uu$ and $\vv$ of norm less than 1 is given
by
\begin{equation} \label{E:Einplus}
\uu \oplus \vv = \frac{1}{1 + \uu^{T} \vv} \left\{ \uu +
\frac{1}{\gamma_{\uu}} \vv + \frac{\gamma_{\uu}}{1 + \gamma_{\uu}}
(\uu^{T} \vv) \uu \right\},
\end{equation}
where $\gamma_{\uu}$ is the well-known \emph{Lorentz factor}
\begin{equation} \label{E:gamma}
\displaystyle \gamma_{\uu} = \frac{1}{\sqrt{1 - \Vert \uu \Vert^2}}.
\end{equation}
Note that $\uu^{T} \vv$ is just the Euclidean inner product of $\uu$
and $\vv$ written in matrix form.

\begin{definition}
The formula (\ref{E:Einplus}) defines a binary operation called
Einstein addition on the open unit ball $\B = \{ \uu : \Vert \uu
\Vert < 1 \}$ of $\R^{3}$.
\end{definition}

The Einstein addition can be naturally defined on the open unit ball
$\B$ of $n$-dimensional space $\R^{n}$. See \cite[Chapter 1]{FS05}
for a derivation of the Einstein addition law.

\begin{remark} \label{R:Loboost}
The Einstein vector addition on the open unit ball $\B$ of $\R^{n}$
can also be defined by applying the Lorentz boost
\begin{displaymath}
B(\uu) = \left(
\begin{array}{cc}
  \gamma_{\uu} & \gamma_{\uu} \uu^{T} \\
  \gamma_{\uu} \uu & I + \frac{\gamma_{\uu}^{2}}{1 + \gamma_{\uu}} \uu \uu^{T} \\
\end{array}
\right)
\end{displaymath}
to $\left(
\begin{array}{c}
  1 \\
  \vv \\
\end{array}
\right)$ and obtaining
\begin{displaymath}
B(\uu) \left(
\begin{array}{c}
  1 \\
  \vv \\
\end{array}
\right) = \left(
\begin{array}{c}
  t \\
  t (\uu \oplus \vv) \\
\end{array}
\right),
\end{displaymath}
where $t = \gamma_{\uu}(1 + \uu^{T} \vv)$.
\end{remark}

If two vectors $\uu$ and $\vv$ are parallel, that is, there exists a
nonzero constant $\lambda$ such that $\vv = \lambda \uu$, then
\begin{displaymath}
\displaystyle \uu \oplus \vv = \frac{\uu + \vv}{1 + \uu^{T} \vv}.
\end{displaymath}

To capture abstractly Einstein addition of velocities in special
relativity, A. A. Ungar has introduced and studied in several papers
and books structures that he has called gyrogroups; see \cite{Un08}
and its bibliography. His axioms are reminiscent of those for a
group, but gyrogroup operations are nonassociative in general.

\begin{definition} \label{D:gyrogroup}
A triple $(G, \oplus, 0)$ is a \emph{gyrogroup} if the following
axioms are satisfied for all $a, b, c \in G$.
\begin{itemize}
\item[(G1)] $0 \oplus a = a \oplus 0 = a$ (existence of identity);
\item[(G2)] $a \oplus (-a) = (-a) \oplus a = 0$ (existence of inverses);
\item[(G3)] There is an automorphism $\gyr[a,b] : G \to G$ for each $a, b\in G$ such that
 \begin{center}
 $a \oplus (b \oplus c) = (a \oplus b) \oplus \gyr[a,b]c$ (gyroassociativity);
 \end{center}
\item[(G4)] $\gyr[0,a] =$ id;
\item[(G5)] $\gyr[a \oplus b, b] = \gyr[a,b]$ (loop property).
\end{itemize}
A gyrogroup $(G, \oplus)$ is \emph{gyrocommutative} if it satisfies
 \begin{center}
 $a \oplus b = \gyr[a,b](b \oplus a)$ (gyrocommutativity).
 \end{center}
A gyrogroup is \emph{uniquely $2$-divisible} if for every $b \in G$,
there exists a unique $a \in G$ such that $a \oplus a = b$.
\end{definition}

We call $\gyr[a,b]$ the gyroautomorphism or Thomas gyration
generated by $a$ and $b$.

\begin{remark} \label{R:gyroloop}
It has been shown that gyrocommutative gyrogroups are equivalent to
Bruck loops \cite{SSS}, i.e, a loop is a Bruck loop if and only if
it is a gyrocommutative gyrogroup with respect to the same
operation. It follows that uniquely $2$-divisible gyrocommutative
gyrogroups are equivalent to $B$-loops, uniquely 2-divisible Bruck
loops.
\end{remark}

A. A. Ungar has shown in \cite[Chapter 3]{Un08} by computer algebra
that Einstein addition on the ball $\B$ is a gyrocommutative
gyrogroup operation, and the gyroautomorphisms are orthogonal
transformations preserving the Euclidean inner product and the
inherited norm. We call $(\B, \oplus)$ the Einstein
(gyrocommutative) gyrogroup, where $\oplus$ is defined by
(\ref{E:Einplus}).

In references \cite{Un08}, \cite{CU1}, and \cite{CU2} the Einstein
vector addition is mostly defined on the open unit ball $\B$ of
$\R^{n}$. It turns our interest to extending the Einstein addition
on the closed unit ball $\overline{\B}$. Substituting $\displaystyle
\alpha_{\uu} = \frac{1}{\gamma_{\uu}}$ in the definition
(\ref{E:Einplus}) we have an alternative expression of the Einstein
vector addition
\begin{displaymath}
\displaystyle \uu \oplus \vv = \frac{1}{1 + \uu^{T} \vv} \left\{ \uu
+ \alpha_{\uu} \vv + \frac{1}{1 + \alpha_{\uu}} (\uu^{T} \vv) \uu
\right\}.
\end{displaymath}
This is well defined for all $(\uu, \vv) \in \overline{\B} \times
\overline{\B}$ excluding the antipodal vectors. Since $\uu \oplus
\vv = \uu$ for any $\vv \in \B$ if $\Vert \uu \Vert = 1$, we are
able to define the Einstein addition naturally for the antipodal
vectors by continuity such that
\begin{displaymath}
\uu \oplus \vv = \uu
\end{displaymath}
for any $\vv \in \overline{\B}$ whenever $\Vert \uu \Vert = 1$.
Hence, the extended Einstein addition of $\uu$ and $\vv$ in
$\overline{\B}$ can be defined as
\begin{displaymath}
\uu \oplus \vv =
\begin{cases}
    (\ref{E:Einplus}) & \hbox{if $\uu \in \B$} \\
    \uu & \hbox{if $\Vert \uu \Vert = 1$} \\
\end{cases}
\end{displaymath}

\begin{remark}
The closed unit ball $\overline{\B}$ with respect to the extended
Einstein addition is a binary system, but not a gyrogroup since it
has no unique inverse.
\end{remark}

\section{Bloch vectors and density matrices}

A qubit is a member of a 2-dimensional Hilbert space, or a two-state
quantum system. A qubit density matrix is a $2$-by-$2$ positive
semidefinite Hermitian matrix with trace $1$. Indeed, any $2$-by-$2$
Hermitian matrix of trace $1$ must have a parametrization of the
form
\begin{displaymath}
\displaystyle \rho_{\vv} = \frac{1}{2} \left(
\begin{array}{cc}
  1 + v_{3} & v_{1} - i v_{2} \\
  v_{1} + i v_{2} & 1 - v_{3} \\
\end{array}
\right),
\end{displaymath}
where
\begin{displaymath}
\vv = \left(
\begin{array}{c}
  v_{1} \\
  v_{2} \\
  v_{3} \\
\end{array}
\right) \in \R^{3}.
\end{displaymath}
So the qubit density matrix can be described as $\rho_{\vv}$ for
some $\vv \in \R^{3}$ such that $\| \vv \| \leq 1$. In this case the
vector $\vv$ is known as the Bloch vector or Bloch vector
representation of the state $\rho_{\vv}$.

\begin{remark} \label{R:eigqubit}
Via the characteristic polynomial of the qubit $\rho_{\vv}$, we
obtain that its eigenvalues are
\begin{displaymath}
\displaystyle \frac{1 + \Vert \vv \Vert}{2}, \ \frac{1 - \Vert \vv
\Vert}{2},
\end{displaymath}
and its determinant is
\begin{displaymath}
\displaystyle \det{\rho_{\vv}} = \frac{1 - \Vert \vv \Vert^{2}}{4}.
\end{displaymath}
So the mixed states are parameterized by the open unit ball $\B$ in
$\R^{3}$.
\end{remark}

The Pauli matrices are given by
\begin{displaymath}
I = \left(
\begin{array}{cc}
  1 & 0 \\
  0 & 1 \\
\end{array}
\right), \sigma_{x} = \left(
\begin{array}{cc}
  0 & 1 \\
  1 & 0 \\
\end{array}
\right), \sigma_{y} = \left(
\begin{array}{cc}
  0 & -i \\
  i & 0 \\
\end{array}
\right), \sigma_{z} = \left(
\begin{array}{cc}
  1 & 0 \\
  0 & -1 \\
\end{array}
\right),
\end{displaymath}
where $i = \sqrt{-1}$. The parameterization of qubit density
matrices can be written in an alternative method using the Pauli
matrices:
\begin{displaymath}
\displaystyle \rho_{\vv} = \frac{1}{2} (I + \vv^{T} \sigma),
\end{displaymath}
where $\vv^{T} \sigma = v_{1} \sigma_{x} + v_{2} \sigma_{y} + v_{3}
\sigma_{z}$ for
\begin{displaymath}
\vv = \left(
\begin{array}{c}
  v_{1} \\
  v_{2} \\
  v_{3} \\
\end{array}
\right), \ \sigma = \left(
\begin{array}{c}
  \sigma_{x} \\
  \sigma_{y} \\
  \sigma_{z} \\
\end{array}
\right).
\end{displaymath}

Let $\D = \{ \rho_{\uu} : \uu \in \B \}$ be the set of all qubit
density matrices representing the mixed states. We define a binary
map $\odot : \D \times \D \to \D$ as
\begin{displaymath}
\rho_{\uu} \odot \rho_{\vv} =
    \displaystyle \frac{1}{\tr{ \left( \rho_{\uu} \star \rho_{\vv} \right) }} \rho_{\uu} \star \rho_{\vv}
\end{displaymath}
where $\displaystyle \rho_{\uu} \star \rho_{\vv} = \rho_{\uu}^{1/2}
\rho_{\vv} \rho_{\uu}^{1/2}$.

\begin{remark} \label{R:Trprod}
From the fact that $\tr{A B} = \tr{B A}$ for any matrices $A$ and
$B$, we have \cite[Eq. (9.67)]{Un08}
\begin{displaymath}
\displaystyle \tr{ \left( \rho_{\uu} \star \rho_{\vv} \right) } =
\tr{\rho_{\uu} \rho_{\vv}} = \frac{1 + \uu^{T} \vv}{2}.
\end{displaymath}
This gives us that the binary map $\odot$ is well defined since $1 +
\uu^{T} \vv \neq 0$ whenever $\uu, \vv \in \B$.
\end{remark}

Since every element in $\D$ is a Hermitian positive definite matrix,
it has a unique square root. We have the explicit form and can prove
it via the direct matrix computation of squaring the form
(\ref{Eq:Squbit}).

\begin{lemma} \label{L:Squbit}
For any $\vv \in \B$ the square root of the qubit density matrix
$\rho_{\vv}$ is uniquely given by
\begin{equation} \label{Eq:Squbit}
\displaystyle \rho_{\vv}^{1/2} = \sqrt{\frac{\gamma_{\vv}}{1 +
\gamma_{\vv}}} \left( \rho_{\vv} + \frac{1}{2 \gamma_{\vv}} I
\right).
\end{equation}
\end{lemma}

We now see an isomorphism between the open unit ball $\B$ with the
Einstein velocity addition $\oplus$ and the binary system $(\D,
\odot)$.

\begin{theorem} \label{T:iso}
The map $\rho : (\B, \oplus) \to (\D, \odot), \ \vv \mapsto
\rho_{\vv}$ is an isomorphism.
\end{theorem}

\begin{proof}
Obviously the map $\rho$ is a bijection. We need to show that
$\rho_{\uu \oplus \vv} = \rho_{\uu} \odot \rho_{\vv}$ for any
\begin{displaymath}
\uu = \left(
\begin{array}{ccc}
  u_{1} & u_{2} & u_{3} \\
\end{array}
\right)^{T}, \vv = \left(
\begin{array}{ccc}
  v_{1} & v_{2} & v_{3} \\
\end{array}
\right)^{T} \in \B.
\end{displaymath}
We set $\displaystyle x := \frac{\gamma_{\uu}}{1 + \gamma_{\uu}}$.
By Lemma \ref{L:Squbit} we have
\begin{displaymath}
\displaystyle \rho_{\uu}^{1/2} = \frac{1}{2} \sqrt{x} \left(
\begin{array}{cc}
  u_{3} + \frac{1}{x} & u_{1} - i u_{2} \\
  u_{1} + i u_{2} & - u_{3} + \frac{1}{x} \\
\end{array}
\right).
\end{displaymath}
It is enough to check the (1,1) and (1,2) entries of $\rho_{\uu}
\odot \rho_{\vv}$ since $\rho_{\uu} \odot \rho_{\vv}$ is Hermitian.

Let us first compute the (1,1) entry of $\rho_{\uu} \star
\rho_{\vv}$. Then
\begin{displaymath}
\begin{split}
\displaystyle & \frac{x}{8} \left\{ \left( u_{3} + \frac{1}{x} \right)(1 + v_{3}) + (u_{1} - i u_{2})(v_{1} + i v_{2}) \right\} \left( u_{3} + \frac{1}{x} \right) \\
& \ \ \ \ \ + \frac{x}{8} \left\{ \left( u_{3} + \frac{1}{x} \right)(v_{1} - i v_{2}) + (u_{1} - i u_{2})(1 - v_{3}) \right\} (u_{1} + i u_{2}) \\
& = \frac{x}{8} \left\{ \left( u_{3} + \frac{1}{x} \right)^{2} (1 + v_{3}) + 2 \left( u_{3} + \frac{1}{x} \right)(u_{1}v_{1} + u_{2}v_{2}) + (u_{1}^{2} + u_{2}^{2})(1 - v_{3}) \right\} \\
& = \frac{x}{8} \left\{ (\Vert \mathbf{u} \Vert^{2} + 2(\uu^{T} \vv)u_{3} - \Vert \uu \Vert^{2} v_{3}) + \frac{2}{x} u_{3} + \frac{2}{x}(\uu^{T} \vv) + \frac{1}{x^{2}} (1 + v_{3}) \right\} \\
& = \frac{x}{8} \Vert \uu \Vert^{2} (1 - v_{3}) + \frac{x}{4} (\uu^{T} \vv) u_{3} + \frac{1}{4} u_{3} + \frac{1}{4} (\uu^{T} \vv) + \frac{1}{8x} (1 + v_{3}) \\
& = \frac{\gamma_{\uu} - 1}{8 \gamma_{\uu}} (1 - v_{3}) + \frac{\gamma_{\uu}}{4(1 + \gamma_{\uu})} (\uu^{T} \vv) u_{3} + \frac{1}{4} u_{3} + \frac{1}{4} (\uu^{T} \vv) + \frac{1 + \gamma_{\uu}}{8 \gamma_{\uu}} (1 + v_{3}) \\
& = \frac{1}{4} \left\{ (1 + \uu^{T} \vv) + u_{3} +
\frac{1}{\gamma_{\uu}} v_{3} + \frac{\gamma_{\uu}}{1 + \gamma_{\uu}}
(\uu^{T} \vv) u_{3} \right\}.
\end{split}
\end{displaymath}
By Remark \ref{R:Trprod} the (1,1) entry of $\rho_{\uu} \odot
\rho_{\vv}$ is
\begin{displaymath}
\begin{split}
& \left( \frac{2}{1 + \uu^{T} \vv} \right) \frac{1}{4} \left\{ (1 + \uu^{T} \vv) + u_{3} + \frac{1}{\gamma_{\uu}} v_{3} + \frac{\gamma_{\uu}}{1 + \gamma_{\uu}} (\uu^{T} \vv) u_{3} \right\} \\
& = \frac{1}{2} \left\{ 1 + \frac{1}{1 + \uu^{T} \vv} \left( u_{3} + \frac{1}{\gamma_{\uu}} v_{3} + \frac{\gamma_{\uu}}{1 + \gamma_{\uu}} (\uu^{T} \vv) u_{3} \right) \right\} \\
& = \frac{1}{2} (1 + (\uu \oplus \vv)_{3}),
\end{split}
\end{displaymath}
where $(\uu \oplus \vv)_{j}$ represents the $j$ th coordinate of
Einstein vector addition $\uu \oplus \vv$.

Similarly, let us compute the (1,2) entry of $\rho_{\uu} \star
\rho_{\vv}$. Then we have
\begin{displaymath}
\begin{split}
& \frac{x}{8} \left\{ \left( u_{3} + \frac{1}{x} \right)(1 + v_{3}) + (u_{1} - i u_{2})(v_{1} + i v_{2}) \right\} (u_{1} - i u_{2}) \\
& \ \ \ \ \ + \frac{x}{8} \left\{ \left( u_{3} + \frac{1}{x} \right)(v_{1} - i v_{2}) + (u_{1} - i u_{2})(1 - v_{3}) \right\} \left( - u_{3} + \frac{1}{x} \right) \\
\end{split}
\end{displaymath}
The real part of the (1,2) entry is
\begin{displaymath}
\begin{split}
& \frac{x}{8} \left\{ \left( u_{3} + \frac{1}{x} \right)(1 + v_{3})u_{1} + (u_{1}^{2} - u_{2}^{2})v_{1} + 2u_{1}u_{2}v_{2} \right\} \\
& \ \ \ \ \ + \frac{x}{8} \left\{ \left( u_{3} + \frac{1}{x} \right) v_{1} \left( - u_{3} + \frac{1}{x} \right) + u_{1}(1 - v_{3}) \left( - u_{3} + \frac{1}{x} \right) \right\} \\
& = \frac{x}{8} \left\{ 2 u_{1} \left( u_{3}v_{3} + \frac{1}{x} \right) + (u_{1}^{2} - u_{2}^{2} - u_{3}^{2}) v_{1} + 2 u_{1}u_{2}v_{2} + \frac{1}{x^{2}} v_{1} \right\} \\
& = \frac{x}{8} \left\{ \frac{2}{x} u_{1} + 2 (\uu^{T} \vv) u_{1} - \Vert \uu \Vert^{2} v_{1} + \frac{1}{x^{2}} v_{1} \right\} \\
& = \frac{1}{4} \left( u_{1} + \frac{1}{\gamma_{\uu}} v_{1} +
\frac{\gamma_{\uu}}{1 + \gamma_{\uu}} (\uu^{T} \vv) u_{1} \right),
\end{split}
\end{displaymath}
and the imaginary part is
\begin{displaymath}
\begin{split}
& \frac{x}{8} \left\{ - \left( u_{3} + \frac{1}{x} \right)(1 + v_{3})u_{2} + (u_{1}^{2} - u_{2}^{2})v_{2} - 2u_{1}u_{2}v_{1} \right\} \\
& \ \ \ \ \ - \frac{x}{8} \left\{ \left( u_{3} + \frac{1}{x} \right) v_{2} \left( - u_{3} + \frac{1}{x} \right) - u_{2}(1 - v_{3}) \left( - u_{3} + \frac{1}{x} \right) \right\} \\
& = \frac{x}{8} \left\{ - 2 u_{2} \left( u_{3}v_{3} + \frac{1}{x} \right) + (u_{1}^{2} - u_{2}^{2} + u_{3}^{2}) v_{2} - 2 u_{1}u_{2}v_{1} - \frac{1}{x^{2}} v_{2} \right\} \\
& = \frac{x}{8} \left\{ - \frac{2}{x} u_{2} - 2 (\uu^{T} \vv) u_{2} + \Vert \uu \Vert^{2} v_{2} - \frac{1}{x^{2}} v_{2} \right\} \\
& = - \frac{1}{4} \left( u_{2} + \frac{1}{\gamma_{\uu}} v_{2} +
\frac{\gamma_{\uu}}{1 + \gamma_{\uu}} (\uu^{T} \vv) u_{2} \right).
\end{split}
\end{displaymath}
By Remark \ref{R:Trprod} the (1,2) entry of $\rho_{\uu} \odot
\rho_{\vv}$ is
\begin{displaymath}
\left( \frac{2}{1 + \uu^{T} \vv} \right) \frac{1 + \uu^{T} \vv}{4}
((\uu \oplus \vv)_{1} - i (\uu \oplus \vv)_{2}) = \frac{1}{2} ((\uu
\oplus \vv)_{1} - i (\uu \oplus \vv)_{2}).
\end{displaymath}
Therefore, we conclude that \cite[Eq. (9.23)]{Un08},
\begin{displaymath}
\displaystyle \rho_{\uu} \odot \rho_{\vv} = \frac{1}{2} \left(
\begin{array}{cc}
  1 + (\uu \oplus \vv)_{3} & (\uu \oplus \vv)_{1} - i (\uu \oplus \vv)_{2} \\
  (\uu \oplus \vv)_{1} + i (\uu \oplus \vv)_{2} & 1 - (\uu \oplus \vv)_{3} \\
\end{array}
\right) = \rho_{\uu \oplus \vv}
\end{displaymath}
for any $\uu, \vv \in \B$. This means the map $\rho$ is a
homomorphism.
\end{proof}

\begin{remark} \label{R:Dgyro}
From Theorem \ref{T:iso}, we obtain that $\D$ is a gyrocommutative
gyrogroup with respect to the operation $\odot$ defined by
\begin{displaymath}
\displaystyle \rho_{\uu} \odot \rho_{\vv} = \frac{1}{\tr{ \left(
\rho_{\uu}^{1/2} \rho_{\vv} \rho_{\uu}^{1/2} \right) }}
\rho_{\uu}^{1/2} \rho_{\vv} \rho_{\uu}^{1/2}.
\end{displaymath}
On $\D$, moreover, the identity is $\displaystyle (1/2) I$ and the
inverse for $\rho_{\uu}$ is
\begin{equation} \label{E:rho}
\displaystyle \rho_{- \uu} = \frac{1}{4 \gamma_{\uu}}
\rho_{\uu}^{-1}.
\end{equation}
\end{remark}

\section{The rapidity metric}

The Einstein gyrogroup on the open unit ball $\B$ admits an analytic
scalar multiplication given by
\begin{equation} \label{Eq:scalar}
(t,\uu) \mapsto t \cdot \uu = \tanh(t \tanh^{-1} \Vert \uu
\Vert)(\uu/\Vert \uu \Vert)
\end{equation}
for $\uu \neq \0$, and $\0$ for $t = 0$ or $\uu = \0$ (see
\cite[Chapter 6]{Un08}). Expressing the magnitude of the velocity
parameter $\uu$ by the hyperbolic parameter $\phi_{\uu}$ called the
\emph{rapidity} of $\uu \in \B$,
\begin{displaymath}
\phi_{\uu} = \tanh^{-1} \Vert \uu \Vert,
\end{displaymath}
we have $t \cdot \uu = \tanh(t \phi_\uu)(\uu/\Vert \uu \Vert)$, or
$\phi_{t \cdot \uu} = t \phi_{\uu}$.

For the Einstein gyrogroup $(\B, \oplus, \0)$, A. A. Ungar considers
what we call the \emph{Ungar gyrometric} defined by
$\varrho(\uu,\vv) = \Vert -\uu \oplus \vv \Vert$. He also defines
what we call the rapidity metric by $d(\uu,\vv) = \tanh^{-1}
\varrho(\uu,\vv)$. It is known as the Cayley-Klein metric on the
Beltrami-Klein model of hyperbolic geometry (see \cite{FS05}), or
the Bergman metric on the symmetric structure $\B$ with symmetries
$S_{\ww}(\vv) = \uu \oplus (- \vv)$ for some $\uu = 2 \cdot \ww$
(see \cite{Zhu}). For real numbers $s, t$, we define
\begin{displaymath}
\displaystyle s \oplus t = \frac{s + t}{1 + st},
\end{displaymath}
the restricted Einstein addition analogous to the Einstein sum of
parallel vectors. We see some properties of rapidity metric on $\B$.

\begin{lemma} \label{L:metric}
The following properties hold for all $\uu,\vv,\ww \in \B$.
\begin{itemize}
\item[(i)] $0 \leq \varrho(\uu,\vv), d(\uu,\vv)$
\item[(ii)] $\varrho(\uu,\vv)=0 \Leftrightarrow d(\uu,\vv)=0 \Leftrightarrow \uu=\vv$
\item[(iii)] $\varrho(\uu,\vv)=\varrho(\vv,\uu)$, $d(\uu,\vv)=d(\vv,\uu)$
\item[(iv)] $\Vert \uu \oplus \vv \Vert \leq \Vert \uu \Vert \oplus \Vert \vv \Vert \Leftrightarrow \varrho(\uu,\ww) \leq \varrho(\uu,\vv) \oplus \varrho(\vv,\ww)
\Leftrightarrow d(\uu,\ww) \leq d(\uu,\vv) + d(\vv,\ww)$
\item[(v)] $\varrho(\uu \oplus \vv,\uu \oplus \ww) = \varrho(\vv,\ww)$ and $d(\uu \oplus \vv,\uu \oplus \ww) = d(\vv,\ww)$
\item[(vi)] $d(\0,r \cdot \ww) = \vert r \vert d(\0,\ww)$.
\end{itemize}
\end{lemma}

We establish the relation of metrics on gyrocommutative gyrogroups
under an injective homomorphism.

\begin{lemma} \label{L:preserve}
Let $(G_1, \oplus, \0)$ and $(G_2, \oplus, \0)$ be gyrocommutative
gyrogroups each equipped with a metric invariant under left
translations. If $\psi : G_1 \to G_2$ is an injective gyrogroup
homomorphism and if there exists $\kappa > 0$ such that
$d(\0,\psi(\uu)) \geq \kappa d(\0,\uu)$ for each $\uu \in G_1$, then
$d(\psi(\uu),\psi(\vv)) \geq \kappa d(\uu,\vv)$ for all $\uu,\vv \in
G_1$.
\end{lemma}

\begin{proof}
Let $\uu,\vv \in G_1$ and set $\ww = -\uu \oplus \vv \in G_1$. By
hypothesis $d(\0,\psi(\ww)) \geq \kappa d(\0,\ww)$. By invariance of
the metrics under left translation $d(0,\ww) = d(\uu \oplus \0,\uu
\oplus (-\uu \oplus \vv)) = d(\uu,\vv)$ and similarly
$$d(\psi(\0),\psi(\ww)) = d(\psi(\uu),\psi(\uu) \oplus \psi(-\uu \oplus \vv)) = d(\psi(\uu),\psi(\vv)).$$
We conclude that $d(\psi(\uu),\psi(\vv)) = d(\0,\psi(\ww)) \geq
\kappa d(\0,\ww) = \kappa d(\uu,\vv)$.
\end{proof}

\section{The trace metric for qubit density matrices}

Here we review the definition of the trace metric on the open convex
cone $\Omega$ of (Hermitian) positive definite matrices. The trace
metric on $\Omega$ is determined locally at the point $A$ by the
differential
\begin{displaymath}
\displaystyle d s = \Vert A^{-1/2} d A A^{-1/2} \Vert_{F},
\end{displaymath}
where $\Vert \cdot \Vert_{F}$ means the Frobenius or Hilbert-Schmidt
norm. This relation is a mnemonic for computing the length of a
differentiable path $\gamma : [a,b] \to \Omega$
\begin{displaymath}
\displaystyle L(\gamma) = \int_{a}^{b} \Vert \gamma^{-1/2}(t)
\gamma'(t) \gamma^{-1/2}(t) \Vert_{F} dt.
\end{displaymath}
Based on the above notion of length, we define the trace metric
$\delta$ between two points $A$ and $B$ in $\Omega$ as the infimum
of lengths of curves connecting them. That is,
\begin{center}
$\displaystyle \delta(A,B) : = \inf \{$ $L(\gamma) :$ $\gamma$ is a
differentiable path from $A$ to $B$ $\}$.
\end{center}

R.\ Bhatia and J.\ Holbrook  have established some basic properties
of the trace metric on $\Omega$ in \cite{BH}.
\begin{lemma} \label{L:Trmetric}
For any $A, B \in \Omega$ and any invertible matrix $X$,
\begin{itemize}
\item[(i)] $\displaystyle \delta(\Gamma_{X}(A), \Gamma_{X}(B)) = \delta(A, B)$, where $\Gamma_{X}$ is a congruence transformation by $X$.
\item[(ii)] $\displaystyle \delta(A, B) = \Vert \log{A^{-1/2} B A^{-1/2}} \Vert_{F}=\Vert \log(A^{-1}B)\Vert_F$.
\item[(iii)] $\displaystyle \delta(A^{1/2}, B^{1/2}) \leq \frac{1}{2} \delta(A, B)$.
\end{itemize}
\end{lemma}

We now consider the trace metric on the set $\D$ of all invertible
qubit density matrices, and then see the connection with the
rapidity metric on the open unit ball $\B$.

\begin{proposition} \label{P:trmetric}
For any $\uu, \vv \in \B$,
\begin{displaymath}
\delta(\rho_{\uu}, \rho_{\vv}) \leq \left\{ \ln^{2} \left(
\frac{x}{a} \right) + \ln^{2} \left( \frac{x}{b} \right)
\right\}^{1/2},
\end{displaymath}
where $\displaystyle x = \frac{1 - \uu^{T} \vv}{2}$, $\displaystyle
a = \frac{1}{4 \gamma_{\uu}^{2}}$, and $\displaystyle b = \frac{1}{4
\gamma_{\vv}^{2}}$.
\end{proposition}

\begin{proof}
By the equation (\ref{E:rho}) and Theorem \ref{T:iso} we have
\begin{displaymath}
\begin{split}
\displaystyle \rho_{\uu}^{-1/2} \rho_{\vv} \rho_{\uu}^{-1/2} & = (2 \gamma_{\uu})^{2} \rho_{-\uu}^{1/2} \rho_{\vv} \rho_{-\uu}^{1/2} \\
& = 4 \gamma_{\uu}^{2} \left( \frac{1 - \uu^{T} \vv}{2} \right)
\rho_{-\uu \oplus \vv}.
\end{split}
\end{displaymath}
We set $\displaystyle K = 4 \gamma_{\uu}^{2} \left( \frac{1 -
\uu^{T} \vv}{2} \right)$ and $\displaystyle p = \frac{1}{2} (1 +
\Vert -\uu \oplus \vv \Vert)$, where $p$ and $1-p$ are the
eigenvalues of $\rho_{-\uu \oplus \vv}$. Then
\begin{displaymath}
\begin{split}
\displaystyle \delta(\rho_{\uu}, \rho_{\vv})^{2} & = \Vert \log{\rho_{\uu}^{-1/2} \rho_{\vv} \rho_{\uu}^{-1/2}} \Vert_{F}^{2} \\
& = \ln^{2}(K p) + \ln^{2} (K(1-p)) \\
& = 2 \ln^{2}{K} + 2 \ln{K} \{ \ln{p} + \ln{(1-p)} \} + \ln^{2}{p} + \ln^{2}{(1-p)} \\
& \leq 2 \ln^{2}{K} + 2 \ln{K} \ln{(p(1-p))} + \ln^{2}{(p(1-p))} \\
& = \ln^{2}{K} + \{ \ln{K} + \ln{(p(1-p))} \}^{2} \\
& = \ln^{2}{K} + \ln^{2}{(K p(1-p))} \\
& = \ln^{2} \left( \frac{x}{a} \right) + \ln^{2} \left( \frac{x}{b}
\right).
\end{split}
\end{displaymath}
The inequality holds since $\ln{p} < 0$ whenever $0 < p < 1$. The
last equality follows from the property $\gamma_{\uu \oplus \vv} =
\gamma_{\uu} \gamma_{\vv} (1 + \uu^{T} \vv)$. Indeed,
\begin{displaymath}
\begin{split}
\displaystyle K p(1-p) & = 4 \gamma_{\uu}^{2} \left( \frac{1 - \uu^{T} \vv}{2} \right) \left( \frac{1 - \Vert -\uu \oplus \vv \Vert^{2}}{4} \right) \\
& = 4 \gamma_{\uu}^{2} \left( \frac{1 - \uu^{T} \vv}{2} \right) \frac{1}{4 \gamma_{-\uu \oplus \vv}^{2}} \\
& = \frac{1}{2 \gamma_{\vv}^{2} (1 - \uu^{T} \vv)} \\
& = \frac{x}{b}.
\end{split}
\end{displaymath}
\end{proof}

\begin{remark}
From Remark \ref{R:Trprod} and Remark \ref{R:eigqubit}, we see that
\begin{displaymath}
\begin{split}
\displaystyle x = \frac{1 - \uu^{T} \vv}{2} & = \tr{ (\rho_{-\uu} \rho_{\vv}) } \\
a = \frac{1}{4 \gamma_{\uu}^{2}} = \det{\rho_{\uu}} \textrm{ and } &
b = \frac{1}{4 \gamma_{\vv}^{2}} = \det{\rho_{\vv}}.
\end{split}
\end{displaymath}
\end{remark}

\begin{theorem} \label{T:Dis}
For any $\uu, \vv \in \B$,
\begin{displaymath}
d(\uu, \vv) \leq \frac{1}{\sqrt{2}} \delta(\rho_{\uu}, \rho_{\vv}).
\end{displaymath}
\end{theorem}

\begin{proof}
We have seen in Remark \ref{R:Dgyro} that $(\D, \odot)$ is a
gyrocommutative gyrogroup via an isomorphism
\begin{displaymath}
\rho : (\B, \oplus) \to (\D, \odot).
\end{displaymath}
By Lemma \ref{L:metric}(v) and Lemma \ref{L:Trmetric}(i) we have the
rapidity metric on $\B$ and the trace metric on $\D$ are both
invariant under left translations. By Lemma \ref{L:Trmetric}(ii) and
Remark \ref{R:eigqubit}
\begin{displaymath}
\begin{split}
\displaystyle \delta((1/2)I, \rho_{\uu})^{2} & = \Vert \log{(2 \rho_{\uu})} \Vert_{F}^{2} \\
& = \ln^{2}{(1 + \Vert \uu \Vert)} + \ln^{2}{(1 - \Vert \uu \Vert)} \\
& \geq \frac{1}{2} \ln^{2}{\frac{1 + \Vert \uu \Vert}{1 - \Vert \uu \Vert}} \\
& = \frac{1}{2} (2 \tanh^{-1} \Vert \uu \Vert)^{2} = 2
d(\0,\uu)^{2}.
\end{split}
\end{displaymath}
So $\delta((1/2)I, \rho_{\uu}) \geq \sqrt{2} d(\0,\uu)$. Therefore,
by Lemma \ref{L:preserve} we conclude
\begin{displaymath}
\delta(\rho_{\uu}, \rho_{\vv}) \geq \sqrt{2} d(\uu,\vv).
\end{displaymath}
\end{proof}

\end{document}